\documentclass{acm_proc_article-sp}

\usepackage{comment,algorithm,algorithmic,url,multirow}
\usepackage{epsfig}

\newtheorem{theorem}{Theorem}

\newtheorem{lemma}{Lemma}

\newdef{definition}{Definition}

\newcommand{\anc}{\preceq}
\newcommand{\desc}{\succeq}

\def\letx{\textbf{let}~}
\def\assert{\textbf{assert}~}

\def\ignore#1{}

\def\dfs{\mathrm{DFS}}
\def\succ{\mathrm{succ}}
\def\live{\mathrm{live}}
\def\lead{\mathrm{lead}}
\def\floor{\mathrm{floor}}
\def\level{\mathrm{level}}
\def\root{\mathrm{root}}

\def\split{\mathrm{split}}
\def\assert{\textbf{assert}~}

\begin{document}

\title{Persistent Cache-oblivious Streaming Indexes}
%
\numberofauthors{1} 
\author{
\alignauthor
Andrew Twigg\thanks{Supported by a Junior Research Fellowship, St Johns College, Oxford}\\
       \affaddr{Oxford University}\\
      \email{adt@cs.ox.ac.uk}
}%
\date{\today}%

\maketitle
\begin{abstract}
In [SPAA2007], Bender et al. define a streaming B-tree (or index) as one that supports updates in amortized $o(1)$ IOs, and present a structure achieving amortized $O((\log N)/B)$ IOs and queries in $O(\log N)$ IOs. We extend their result to the partially-persistent case. For a version $v$, let $N_v$ be the number of keys accessible at $v$ and $N$ be the total number of updates. We give a data structure using space $O(N)$, supporting updates to a leaf version $v$ with $O((\log N_{v})/B)$ amortized IOs and answering range queries returning $Z$ elements with $O(\log N_{v} + Z/B)$ IOs on average (where the average is over all queries covering disjoint key ranges at a given version). This is the first persistent `streaming' index we are aware of, i.e. that supports updates in $o(1)$ IOs and supports efficient range queries.
\end{abstract}

\terms{Cache-oblivious algorithms, External-memory algorithms, Versioned data structures}

\section{Introduction}

Indexes (such as B-trees) are fundamental to many problems in storage, such as databases and file systems. In this paper we investigate IO-efficient persistent indexes that have fast updates (in $o(1)$ IOs per update), linear space and support efficient range queries.

Logically, there is a tree of versions $V$, and every version $v \in V$ has a dictionary $D_u$ mapping keys to values. We want to support the following operations, starting from an empty data structure and a single root node:
\begin{itemize}
\item $\texttt{update}_v (k,x)$: create a new child $w$ of $v$ with $D_w = D_v \cup \{ (k,v,x) \}$ (overwriting any previous elements for $k$). Return $w$.
\item $\texttt{query}_v (k_1,k_2)$: return $\{ (k,v,x) \in D_v : k \in [k_1,k_2] \}$. 
\end{itemize}

If $\texttt{update}_v(\cdot)$ only works on leaf versions $v$, then the structure is partially-persistent (the version tree is a line), otherwise it is fully-persistent. This work focuses on the partially-persistent case. In addition, we only the discuss the case of updates; deletes can be handled by inserting a key with a null value; when these values are encountered in the output of a query, we can ignore the element.

We say that a key $k$ is \emph{live} at $v$ if $k \in D_v$. We let $N_v = |D_v|$ (the quantity increases down the version tree) and $N$ be the total number of updates. Clearly, storing the dictionaries $\{D_v\}$ explicitly is not efficient - the description is purely to simplify the exposition. 

\subsection{Related work}
Unversioned indexes have been developed which support a range of tradeoffs between update and query performance. The cache-oblivious (CO) model was introduced by Frigo et al. \cite{Frigo:1999}. Several update/query tradeoffs are known for unversioned dictionaries in the CO model. The cache-oblivious lookahead array (COLA) of Bender et al. \cite{Bender:2007} supports updates in amortized $O(\log N/B)$ IOs and queries in $O(\log N + Z/B)$ IOs. The xDict structure of Brodal et al.\cite{Brodal:2012} supports a wide range of update/query tradeoffs in the  CO model.

The classic versioned analogue of the B-tree is the copy-on-write (CoW) B-tree \cite{episode,Rodeh:2008}, which 
is based on the path-copying technique of Driscoll et al. \cite{dsst} for making
pointer-based data structures fully-persistent. The structure has three problems: 1) each update may cause a new path to be written, giving $\Theta(N B \log_B N)$ space; 2) updates cost $O(\log_{B} N)$ IOs, and 3) it is not cache-oblivious.

Becker \cite{96-becker-mvbt} presented the multiversion B-tree (MVBT), which solves the space blowup problem. It achieves the same query and update bounds with $O(N)$ space, and is partially-persistent. Lanka et al. \cite {115861} developed two fully-persistent B-tree variants based on variants of the techniques from \cite{dsst}, but for both variants, either there is a large space blowup, or range queries may be far from optimal.
Recently, Brodal et al.~\cite{Brodal:fullypersistent} presented a fully-persistent B-tree that uses linear space, and achieves roughly the same update/query bounds (within a small factor) as the MVBT and the CoW B-tree. These results are summarized in Table \ref{tbl:comparison}.

Several solutions are known for offline (or batched) constructions of persistent B-trees. Arge et al. \cite{Arge:2003} use a modified MVBT to solve IO-efficient planar point location. Goodrich et al. \cite{Goodrich:1993} give an offline method for constructing a persistent B-tree using $O((N/B) \log_{M/B} (N/B))$ IOs (where $M$ is the size of the memory in the DAM model). It is not clear if any of the offline constructions can be made cache-oblivious.

Our structure is the first `streaming' persistent index we are aware of, i.e. that supports updates in $o(1)$ IOs and supports efficient range queries.

\begin{table*}
\center
\begin{tabular}{|c|c|c|c|c|c|c|}
\hline
\bf Result & \bf Cache-oblivious? & \bf Persistence & \bf Update & \bf Range query (size $Z$) & Space\\
\hline
B-tree \cite{1972-bayer-btree} & No & None & $O(\log_{B} N)$ & $O(\log_{B} N + Z/B)$ & $O(N)$ \\
CoW B-tree \cite{episode,Rodeh:2008} & No & Full & $O(\log_{B} N_{v})$ & $O(\log_{B} N_{v} + Z/B)$ & $O(N B \log_{B} N)$ \\
MVBT \cite{96-becker-mvbt} & No & Partial & $O(\log_{B} N_{v})$ & $O(\log_{B} N_{v} + Z/B)$ & $O(N)$ \\
Lanka et al. \cite{115861} & No & Full & $O(\log_{B} N_{v})$ & $O((\log_{B} N_{v})(1+Z/B))$ & $O(N)$ \\
Brodal et al. \cite{Brodal:fullypersistent} & No & Full & $O^{*}(\log_{B} N_{v} + \log_{2} B)$ & $O(\log_{B} N_{v} + Z/B)$ & $O(N)$ \\
COLA \cite{Bender:2007} & Yes & None & $O^{*}((\log N)/B)$ & $O(\log N + Z/B)$ & $O(N)$ \\
This paper & Yes & Partial & $O^{*}((\log N_{v})/B)$ & $O^{*}(\log N + \log^2 N_{v} + Z/B)$ & $O(N)$ \\
\hline
\end{tabular}
\caption{Comparing related work. Bounds marked $O^{*}(\cdot)$ are amortized, or average-case over operations on a given version. Only the last two structures are `streaming', i.e. support updates in $o(1)$ IOs.}
\label{tbl:comparison}
\end{table*}

Afshani et al. \cite{AHZ.LB.SOCG09} prove that any partially-persistent CO index that answers queries in $O(\log_B N + Z/B)$ IOs must use superlinear space. Our query bound comes close to this, except that the first term in our bound is $O(\log^2 N_v)$, and the bound is an average-case bound, not worst-case. It remains open to improve our query bound to a worst-case one.


\subsection{Our results}

We present a cache-oblivious, partially-persistent structure with the following properties:
\begin{itemize}
\item updates to a leaf version $v$ cost $O(\log N_v / B)$ amortized IOs
\item the structure uses space $O(N)$ for $N$ updates
\item a range query at version $v$ returning $Z$ elements costs $O(\log N + \log^2 N_v + Z/B)$ IOs on average (where the average is taken over queries for disjoint keys at version $v$).
\end{itemize}

In order to achieve these bounds, we must manage a tradeoff between duplicating enough keys so that range queries can be fast, and not duplicating too many so that we can retain fast updates and linear space. A novel part of our construction is a collection of exponentially-growing `versioned' arrays that simultaneously have a lower bound on `density' (the fraction of keys live at some version, see Section \ref{sec:density}) and an upper bound on the number of elements copied from other arrays.

\subsection{Structure of paper}

In Section \ref{sec:structure} we describe the data structure and the auxiliary structures to make it work. In Section \ref{sec:operations} we describe the update and query operations. In Section \ref{sec:analysis}, we prove the update, query and space bounds. We conclude by mentioning some open problems.

\section{The data structure}
\label{sec:structure}

In this section we describe the components of the structure, in order to describe the operations in the next section. At a high level, the structure consists of a collection of versioned arrays (arrays of elements each having a different version set attached to them), organized into levels, where the sizes of arrays roughly doubles between levels. Each level may have many arrays of roughly the same size, but each of them are tagged with disjoint version sets. In this way, a range query only needs to examine one array per level. Arrays are promoted to the next level when they become too large, and roughly speaking, arrays with overlapping version sets in the same level are merged together. In order to maintain the range query performance, we need (as in most of the persistent data structures) to duplicate some elements. We use the notions of live, lead and density to track the number of duplicated elements in arrays. When an array has too many replicated elements, we subdivide it into several smaller arrays, each of which has not too many replicated elements. This way, we control the balance between replicating enough elements for good range query performance and not replicating too much so that we get good update and space bounds. The rest of the section gives the details.

\subsection{Versioned arrays}
Elements are (key, version, value) tuples, which we often write as $(k,v,x)$, and sometimes we omit the value for simplicity. A versioned array $(A,W)$ stores a list of elements $(k,v,x)$ and a set of versions $W$, where every element $(k,v,x)$ is live for some version $w\in W$, and where $W$ is a connected subtree of $V$. Within an array $(A,W)$, elements are ordered lexicographically by $(k,v)$, where keys have some natural ordering and versions are ordered \emph{decreasing} by their DFS number in $W$ (so that each array has its own DFS numbering).

With this structure, we can search within an array as follows. For a key $k$, all the descendants of version $v$ form a contiguous region to the left of $v$ in the array. Hence $\texttt{query}_v(k_1,k_2)$ can be performed by first binary searching for the first element with key $k$ and then scanning to the right for the first element $(k,w,x)$ where $w$ is an ancestor of $v$; whis is the closest ancestor to $v$ for this key. We then continue scanning until we find the next key, until we have covered all keys in the range $[k_1,k_2]$.

There are two concerns with this searching method. First, it might be inefficient if there are lots of irrelevant (non-live) elements for the query version. We will deal with this by requiring that all arrays have some constant density bound, defined below. Second, we need a way to test ancestorship quickly for each element, without doing an IO per element. We use the following (well-known) method. Let $\dfs_W(x)$ be the DFS number of $x$ in the version tree $W$. Then for versions $w,v$, we have $w \anc v \iff \dfs(v) \in I(w)$. For each element $(k,w)$ in an array $(A,W)$, we store (alongside the element) the interval $$I(w) = [\dfs_W(w), \max_{x \desc w} \dfs_W(x)].$$  Now, for a query at version $v$ and an array $(A,W)$, if we know $\dfs_W(v)$, we can check ancestorship of elements in $(A,W)$ with no additional IO. Note that we still have to deal with efficiently figuring out $\dfs(v)$ on a query for $v$. For the partially-persistent case, this is straightfoward to encode into the version numbers - version $v_i$ has DFS $i$ if it was the $i$th version created. It is not clear how to do this efficiently for the fully-persistent case.

\subsection{Live, density, lead}
\label{sec:density}

An element $(k,v,x)$ is \emph{live} at $w$ if $(k,v,x) \in D_v$. For an array $(A,W)$, we let $\live(A,w)$ be the number of elements of $A$ live at $w$, and $\live(A,W) = \sum_{w \in W} \live(A,w)$. Array $(A,W)$ has \emph{density} $\delta(A,W) = \min_{w \in W} \live(A,w) /  |A|$.

An element $(k,v,x)$ is \emph{lead} at exactly one version $v$ (so that the lead elements are the `original' elements and the live elements are the `inherited' copies). We define $\lead(A,w)$ and $\lead(A,W)$ similarly to the live quantities. Array $(A,W)$ has \emph{lead fraction} $\sum_{w \in W} \lead(A,w) /  |A|$.

For an array $(A,W)$ and version $v \in W$, let $S(A,v)$ be a subarray containing all the elements of $A$ live at some descendant of $v$ in $W$.

The importance of density is the following. If density is too low, then scanning an array to answer a range query at version $v$ involves reading many elements not live at $v$. If we insist on density being too high, then many elements must be duplicated, leading to a large space blowup (in the limit, each array will contain live elements for a single version, i.e. the dictionaries $D_v$). Our construction guarantees that every array has density at least $1/6$, and almost all arrays have lead fraction at least $1/3$.

\subsection{Levels of arrays}
The data structure consists of a series of exponentially-growing arrays, organized into levels, as in the COLA \cite{Bender:2007}.  Every level $l$ may contain many arrays, where each array $(A,W)$ satisifies:
\begin{enumerate}
\item $|A| \le 2^{l+1}$
\item $\live(A,w) \ge 2^l / 3$, for all $w \in W$
\end{enumerate}
Note that these together imply that every array has density at least $1/6$. In addition, the sets $W_i$ appearing a given level are all disjoint. This means that, for each level, we only need to examine a single array to answer a query.

\subsection{Auxiliary structures}
\label{sec:prelim:aux}
We also keep an auxiliary index for each level that maps from versions to arrays in that level. For each array $(A,W)$ at level $l$, since $W$ is a connected subtree of $V$, we can describe it by storing an interval with its topmost and bottommost versions. Call this the interval for $W$. In every level, we store a COLA \cite{Bender:2007} on these intervals, keyed by the left part of the interval.

A lookup at level $l$ can query this structure in $O(\log N)$ IOs to determine which array to examine. The structure is modified only when some array $(A,W)$ is merged or rewritten at level $l$; in this case, at most $O(|W|)=O(|A|)$ elements may be modified in the structure.

We also store, in a separate structure, for every version $v$, the highest level where there is an element live at $v$. This allows us to only examine the $O(\log N_v)$ levels where there might be an array, instead of checking all the $O(\log N)$ levels.

In Section \ref{sec:faster} we show how to remove the need for a per-level auxiliary structure.

\section{Operations}
\label{sec:operations}


\subsection{Query}
\label{sec:query}

To answer $\texttt{query}_v(k_1,k_2)$, we do the following: for each level $l$, find the unique array $(A_l,W_l)$ that should be examined. As described above, binary search for the first element for key $k_1$ and scan to the right, reporting for each key $k$ in the range, the element for the first version $w$ that is an ancestor of $v$ we encounter. The reports from each level are then merged, and we finally report the closest ancestor version $w$ to $v$ for each key in the range. This final merge can of course be done while scanning each array in parallel so that we don't store the whole intermediate output.

\begin{algorithm}[!h]
\caption{\texttt{query}$_v(k1,k2)$}
\label{alg:query}
\begin{algorithmic}[1]
\FOR {each level $l$}
	\STATE \letx $(A,W)$ be the array with $v \in W_l$
	\STATE binary search for $k_1$ in $A$
	\STATE \letx $S_l = []$
	\FOR {each key $k$ in $[k_1,k_2]$}
		\STATE add to $S_l$ the first element $(k,y)$ where $anc(y,v)$
	\ENDFOR
\ENDFOR
\STATE \letx $S$ = merge$(S_1...S_l)$
\FOR {each key $k$ in $S$}
	\STATE keep only the first element $(k,w,x)$ of $S$ // the closest ancestor to $v$
\ENDFOR
\RETURN $S$
\end{algorithmic}
\end{algorithm}

\subsection{Update}

An update operation $\texttt{update}_v(k,x)$ is the promotion of a singleton array with a new version into level 0, by calling \texttt{promote}$(\{(k,v+1,x)\},v+1,0)$. The \texttt{promote} algorithm is shown in Algorithm \ref{alg:promote}.

\begin{algorithm}[h]
\caption{\texttt{promote}$(A,W,l)$}
\label{alg:promote}
\begin{algorithmic}[1]
\STATE $(B,Y)=$ array at level $l$ where $Y$ has the closest ancestor to $\min(W)$
\STATE $(A',W')=(A \cup B, W \cup Y) = (A \cup B, [\min(Y),\max(W)])$
\STATE register $(A',W')$ at level $l$
\IF {$|A'| > 2^{l+1}$}
	\STATE \letx $(A'',W'') =$ \texttt{extract\_promotable}$(A',W',l)$
	\STATE \texttt{promote}$(A'',W'',l+1)$
	\STATE rewrite the array $(A',W') := (A' \setminus A'', W' \setminus W'')$
\ENDIF
\IF {$\delta(A',W') < 1/6$}
	\STATE \letx ${(A_1,W_1)...(A_k,W_k)}$ = \texttt{subdivide}$(A',W',l)$
	\FOR {each $i$}
		\STATE register $(A_i,W_i)$ at level $l$
	\ENDFOR
	\STATE unregister $(A',W')$ at level $l$
\ENDIF
\end{algorithmic}
\end{algorithm}

We first find the array $(B,Y)$ that the incoming array $(A,W)$ will merge with. We select the array whose version set contains the closest ancestor to $\min(W)$. In line 2, we merge the two arrays to get a new array $(A \cup A', W \cup W')$.

As a result of the merge, the new array may violate the size constraint at level $l$. In this case (lines 4-8), we search for and extract a promotable subarray -- that is, the largest subarray satisfying the conditions to exist at level $l+1$ -- and promote it to level $l+1$, if one exists. The remainder array is left at level $l$. This is shown in Algorithm \ref{alg:extract_promotable}. Note that the $S(\cdot)$ array extracted, if there is one, may contain duplicate 


\begin{algorithm}[h]
\caption{\texttt{extract\_promotable}$(A,W,l)$}
\label{alg:extract_promotable}
\begin{algorithmic}[1]
\STATE \letx $v$ be the highest version with $\live(A,v) > 2^{l+1}/3$ and
 $|S(A,v)| > 2^{l+1}$
\RETURN $S(A,v)$, or {\bf null} if no such $v$ exists
\end{algorithmic}
\end{algorithm}

The remainder array may not satisfy all the conditions for level $l$ - there may be a version $v$ that was previously dense in the smaller array pre-merge, but is no longer dense in the larger array post-merge. In this case (lines 9-15), we greedily subdivide the array into a collection of arrays $(A_i,W_i)$, each of which we will show later is dense and satisifes all the constraints for level $l$. The procedure is shown in Algorithm \ref{alg:subdivide}. 

\begin{algorithm}[h]
\caption{\texttt{subdivide}$(A,W,l)$}
\label{alg:subdivide}
\begin{algorithmic}[1]
\STATE \letx $C = \{\}$
\STATE \letx $v = \min(W)$
\WHILE {$W \ne \emptyset$}
	\STATE \assert $v$ is not a leaf
	\STATE \letx $w=child(v)$
	\IF {$|S(A,w)| >= 2^{l+1}$}
		\STATE $w := child(v)$
	\ELSE
		\STATE \letx $X$ be all descendants of $w$ in $W$
		\STATE add $(S(A,w),X)$ to $C$
		\STATE $W := W \setminus X$
		\STATE $v := \min(W)$
	\ENDIF
\ENDWHILE
\RETURN $C$
\end{algorithmic}
\end{algorithm}

Note that after this promotion and subdivision, some elements may be duplicated among arrays - these are the `live' elements kept around to ensure good range query performance. 

\section{Analysis}
\label{sec:analysis}

\begin{lemma}[Promotion]
\label{lem:promotion}
Consider an array $(A,W)$ promoted from level $l$ to $l+1$. It satisfies (1) $\live(A,v) \ge 2^{l+1}/3$ for all $v \in W$; (2) it contains at least $(2/3) 2^{l+1}$ lead elements; (3) the array has density $\ge 1/6$.
\end{lemma}
\begin{proof}
For readability, we drop the array $A$ in the terms $S(), \live, \lead$ etc. unless we specify otherwise. 

Condition (1): follows directly from Algorithm \ref{alg:extract_promotable}.

Condition (2): let $v$ be the highest version where $\live(A,v) > 2^{l+1}/3$ and $|S(A,v)| > 2^{l+1}$. Let $U$ be the set of versions in $W$ descendant from $v$. Since $v$ was the first element that satisfied the promotion critera, we also have (letting $p(v)$ be the parent of $v$) $\live(p(v)) < 2^{l+1}/3$. Considering the elements of $A$ lead at $U$, we have
\begin{eqnarray*}
\lead(U) &=& |S(p(v))| - \live(p(v)) \\
&\ge& |S(v)| - \live(p(v)) \\
&\ge& 2^{l+1} - 2^{l+1}/3\\
&=& (2/3) 2^{l+1}.
\end{eqnarray*}

Condition (3): condition (1) and since the array has size $\ge 2^{l+2}$ imply that the density is at least $1/6$.
\qed
\end{proof}

The main result for subdivision is the following.
\begin{lemma}[Subdivision]
\label{lem:subdivision}
Consider the remainder $(A,W)$ after extracting all subarrays $(A',W')$ having $|A'| \ge 2^{l+1}$ and $\live(A',v) \ge 2^{l+1}/3$ for all $v\in W'$. Algorithm \ref{alg:subdivide} outputs arrays $(A_i, W_i)$ where (1) $|A_i| < 2^{L+1}$; (2) all arrays except at most one have lead fraction at least $2/3$; (3) all arrays have density $\ge 1/6$.
\end{lemma}
\begin{proof}
For readability, we drop the array $A$ in the terms $S(), \live, \lead$ etc. unless we specify otherwise. 

Condition (1): follows directly by looking at the algorithm.

Condition (2): from the conditions of the lemma, we have that for all $v\in W$, either $|S(v)| < 2^{l+1}$ or $\live(v) < 2^{l+1}/3$. Let $v$ be the highest version where $|S(v)| < 2^{l+1}$, with parent $p(v))$. Then $|S(p(v))| \ge 2^{l+1}$ and $\live(p(v)) < 2^{l+1}/3$.

Let $U$ be the descendants of $v$ in $W$, then similarly to the proof of Lemma \ref{lem:promotion}, we have$$\lead(U) = |S(p(v))| - \live(p(v)) > (2/3) 2^{l+1}.$$ Since $|S(v)| < 2^{l+1}$ and $\lead(U) > (2/3) 2^{l+1}$, this implies that the array $(S(v),U)$ output has lead fraction at least $2/3$. The preconditions hold after each iteration of the algorithm, since $|S(v)|$ and $\live(v)$ can only decrease by extracting some subarray.

Condition (3): from Lemma \ref{lem:promotion}, all promoted arrays $(A,W)$ satisfy $\live(A,v)>2^l/3$ for all $v\in W$, and the quantity $\live(\cdot)$ only increases down the version tree. Combined with (1), all the output arrays have density $\ge 1/6$. 

The only remaining thing that could go wrong is that we hit a leaf node when walking down the tree. We will show that this cannot happen. Assume we reach a leaf $v$. Then $|S(A,v)| \ge 2^{l+1}$ (*). Hence $\live(A,v) < 2^{l+1}/3$ and so $\live(p(v)) < 2^{l+1}/3$. Since $v$ is a leaf, we have $\live(v)=\lead(v)$, so
\begin{eqnarray*}
|S(A,v)| &\le& \lead(A,v) + \live(A,p(v)) \\
&\le& \live(A,v) + \live(A,p(v)) \\
&\le& (2/3) 2^{l+1}.
\end{eqnarray*}
This contradicts (*) above.  
\qed
\end{proof}

Together, Lemmas \ref{lem:promotion} and \ref{lem:subdivision} imply that all arrays have density at least $1/6$.

\subsection{Update bound}
Assume we have a memory buffer of size at least $B$. Each array involved in a merge has size at least $B$, so a merge of some number of arrays of total size $k$ elements costs $O(k/B)$ IOs. In the unversioned COLA \cite{Bender:2007}, each element exists in exactly one array and may participate in $O(\log N)$ merges. One difficulty here is that an element may exist in many arrays, and may also participate in many merges at the same level (e.g., when an array at level $l$ is subdivided and some subarrays repeatedly remain at level $l$). We shall prove the result using an accounting argument by charging a merge to the lead elements in the promoted array that triggered the merge.

\begin{theorem}
\label{thm:update1}
The operation $\texttt{update}_v(\cdot)$ costs amortized $O((\log N \log N_v) / B)$ cache-oblivious IOs.
\end{theorem}
\begin{proof}
Assume that each IO costs $\$1$ and can read/write a block of $B$ elements. When $(A,W)$ is promoted, all its lead elements are given credit $\$c/B$, for some constant $c>0$. We will show that a $c$ exists so that the lead elements of $(A,W)$ can pay for all the IOs triggered at this level by the promotion of $(A,W)$ (this will be enough for an inductive argument).

By Lemma \ref{lem:subdivision}, every output array $(A_i,W_i)$ has lead fraction $\ge 1/3$ except at most one array that has size $\le 2^{l+1}$. So the total output size is at most
\begin{eqnarray*}
&& 3 \sum_i \lead(A_i,W_i) + 2^{l+1} \\
&\le & 6 . \lead(A,W) + 2^{l+1} \\
&\le & 6 |A| + 2^{l+1}\\
&\le & 7.2^{l+1} \\
\end{eqnarray*}
where the second line follows since every lead element in $(A,W)$ ends up in exactly one output array $(A_i,W_i$). By Lemma \ref{lem:promotion}, $(A,W)$ has at least $(2/3) 2^l$ lead elements. Therefore, choosing $c>21$ suffices.

Since every array has density $\ge 1/6$, any array $(A,W)$ with $v \in W$ has size at most $6 N_v$, so it cannot exist in a level higher than $6 \lceil \log n \rceil$. Hence each inserted element will be charged in total at most $\$O(c (\log N_v)/B)$.

We now need to account for updating the per-level auxiliary structures as described in \ref{sec:prelim:aux}, which will contribute an additional $O(\log N / B)$ factor per update. When an array $(A,W)$ is promoted into a level, at most $O(|W|)=O(|A|)$ update operations are made on the auxiliary structure (noting that we do not need to perform any reads, just updates, which may overwrite existing elements), which adds a total of $O((|A| \log N)/ B)$ IOs for each level.
\qed
\end{proof}

\subsection{Space bound}

\begin{theorem}
\label{thm:space}
The data structure requires space $O(N)$.
\end{theorem}
\begin{proof}
The proof of the update bound showed that whenever an array containing $k$ lead elements is promoted, at most $O(k)$ space is used. Each lead element gets promoted at most once per level, and the number of lead elements per array doubles between successive
promotions. Thus the total space used for arrays at levels $l>0$ 
is $O(\sum_{i>0} N/2^i) = O(N)$. \qed
\end{proof}


\subsection{Query bound}

To answer a query $\texttt{query}_v(k_1,k_2)$, we find at each level $l$ a unique array $(A,W)$ as described in \ref{sec:query}, binary search to find the right starting position, then merge the relevant portions of these arrays, only reporting the elements with version being the closest ancestor to $v$. If a key has elements $E$ in multiple arrays, we first discard any versions that have descendant elements in $E$, and break remaining ties (if any) by taking element in the lowest-level array.

Since all arrays have density $\ge 1/6$, this immediately guarantees that `voluminous queries' (that report a constant fraction of all elements live at version $v$) are asymptotically optimal - they perform $O(N_v / B)$ IOs. A small range query may be forced to scan a large part of an array, but we will show that the average cost, taken over all disjoint key ranges, will be efficient. Equivalently, a random key range will be efficient in expectation.

\begin{lemma}
\label{lem:query}
Consider querying an array $(A,W)$ for a version $v \in W$. The average cost of a range query that returns $Z$ elements is $O(\log N_v + Z/B)$ IOs, where the average is taken over all disjoint keys from $A$ that are live at $v$.
\end{lemma}
\begin{proof}
Let $\Sigma = (x_1, y_1), (x_2, y_2), \ldots$ be a set of key-disjoint ranges. For a range $\sigma \in \Sigma$, 
let $f_v(A,\sigma)$ be the number of IOs used in examining elements in $A$ for the query $\texttt{query}_v(\sigma)$.
Since the elements of $A$ are ordered lexicographically by $(k,v)$, the regions of $A$ 
examined by each key range $\sigma$ are disjoint, hence $\sum_{\sigma} f_v(A,\sigma) \le |A|$.
By density, we have $|A| \le 6 .\live(A,v)$, so the total cost is bounded by $O(\live(A,v)) = O(N_v)$. All the reported keys are live, and disjoint, so the average bound follows, with the $O(\log N_v)$ term from binary searching for the first element to scan.
\qed
\end{proof}

\begin{theorem}
\label{thm:query1}
The average cost of a query for version $v$ returning $Z$ keys is $O(\log N \log^2 N_v + \log N + Z/B)$ IOs, where the average is as in Lemma \ref{lem:query}.
\end{theorem}
\begin{proof}
For each level, we query an auxiliary search structure of size $O(N)$ to determine which array to examine at that level, which costs $O(\log N)$. It costs $O(\log N_v)$ to binary search each array, and there are at most $O(\log N_v)$ arrays that need to be examined, at most one per level. This gives the first term.
Let $A_l$ be the unique array examined at level $l$. Summing Lemma \ref{lem:query} over all such arrays, the total number of elements examined is
$$\sum_{\sigma} \sum_{l=0}^{O(\log N_v)} f_v(A_l,\sigma) \le \sum_{l=0}^{O(\log N_v)}{|A_l|} = O(N_v).$$
So on average over all key-disjoint queries covering live elements for version $v$, a query reporting $Z$ keys touches $O(Z)$ elements over all arrays it examines.
\qed
\end{proof}

\subsection{Faster queries and updates}
\label{sec:faster}

We can improve the bounds by removing the need for the auxiliary structures described in \ref{sec:prelim:aux}. We do this by exploiting some more structure between the arrays in various levels.

The idea is to define, for every array $(A,W)$, a unique successor array $\succ(A,W)$ in the next level, so that after we have queried $(A,W)$, we examine the successor pointer and query $\succ(A,W)$. Of course, the lowest level may contain as many arrays as versions, so that we will still need to do a search to find the appropriate starting array in level 0.

Consider some version $w$. When an array $(A,W)$ with $w \in W$ is promoted from level $l$ to $l+1$, it contains a copy of all the live elements at $w$ that are in some array $(A',W')$ at level $l$. Later, when a new array with a version $w$ is promoted into level $l$, the merge procedure from Algorithm \ref{alg:promote} will merge it with the array $(A',W')$. Of course, this is not a problem from a correctness perspective, but it means that we cannot define a unique successor array; many such arrays may merge into a single array at some intermediate level and be promoted.

 The crucial thing to note is that the arrays in higher levels already have all the required live elements to answer a query correctly. The following definition captures when an array already exists for that version at a higher level.

{\bf Definition:} For a version $w$, let $\level(w)$ be the lowest level that contains an array $(A,W)$ with $w \in W$. For an array $(A,W)$ at level $l$, let $\floor_l(W)$ be the closest ancestor $w$ to $\root(W)$ such that $\level(w) > l$. Note that $\floor_l(W)$ may not exist if $(A,W)$ is the highest such array; in this case we let $\floor_l(W) = \mathrm{null}$.

We modify the merge procedure as follows: when promoting from level $l$ to $l+1$, merge $(A,W)$ with the unique array $(A',W')$ where $\floor_l(W) \in W'$. If no such array exists (or if $\floor_l(W)$ is null) then leave $(A,W)$ in place (in level $l+1$). For an array $(A,W)$ in level $l$, let $\succ(A,W)$ be the unique array in level $l+1$ that $(A,W)$ would merge with in this way. We can arrange it so that $\succ(\cdot)$ always exists, by leaving `dummy' arrays in a level after an array has been promoted out of that level. If $(A,W)$ is the highest such array, then we let $\succ(A,W)=\mathrm{null}$ so that we know when to terminate the search.

The successor pointers eliminates the need to query the auxiliary structure once per level, and improves updates by avoiding having to write out all the new $W$ entries when an array $(A,W)$ is promoted. This gives the following result.

\begin{theorem}
With successor pointers, the average cost of a query for version $v$ returning $Z$ keys is $O(\log N + \log^2 N_v + Z/B)$ IOs, where the average is as in Lemma \ref{lem:query}. The amortized update cost at version $v$ is $O((\log N_v)/B)$ IOs, where the amortization is as in Theorem \ref{thm:update1}.
\end{theorem}

\section{Open problems}

The lower bound of Afshani et al. \cite{AHZ.LB.SOCG09} does not preclude achieving a worst-case query bound of $O(log N_{v} + Z/B)$ IOs. It would be interesting to try to achieve this bound, or at least to make the bound we achieve here a worst-case one. The notion of density seems too weak to achieve this, as it only considers whole arrays. We have tried to extend it to a more `local' notion, but without much success.

Another open problem is to make our results fully persistent. The subdivision procedure can be extended without much difficulty to the fully-persistent case (by considering subsets of child subtrees to include in the output arrays), but the limiting factor is currently that we do not know how to efficiently maintain the DFS numbers for the version tree with $o(1)$ IOs per update and query cost $O(\log N)$ IOs; in the partially-persistent case this is not a problem since we can use the creation order of versions. Indeed, being able to solve the closely-related problem of maintaining the subtree sizes of a tree subject to inserts only, with the same update/query bounds, also appears to be a challenging problem.

\bibliographystyle{plain}
\bibliography{references}

\end{document}